\documentclass[%onecolumn,%showpacs,eqsecnum,preprintnumbers,
twocolumn,amsmath,amssymb,superscriptaddress,10pt
]{revtex4}
\usepackage{bm}
\usepackage{graphicx}
\usepackage{amsbsy} 
\usepackage{amsmath}
\usepackage{amsfonts}
\usepackage{amsthm}
\usepackage{bm}
\usepackage{graphicx}
\usepackage{amsbsy}
\usepackage{amsmath}
\usepackage{amsfonts}
\usepackage{amsthm}
\usepackage{braket}
\usepackage{color}
\usepackage{mathrsfs}

\begin{document}

\theoremstyle{plain}
\newtheorem{theorem}{Theorem}
\newtheorem{lemma}[theorem]{Lemma}
\newtheorem{corollary}[theorem]{Corollary}
\newtheorem{conjecture}[theorem]{Conjecture}
\newtheorem{proposition}[theorem]{Proposition}

\theoremstyle{definition}
\newtheorem{definition}{Definition}
\newtheorem{algorithm}{Algorithm}
\theoremstyle{remark}
\newtheorem*{remark}{Remark}
\newtheorem{example}{Example}

\title{Recursive QAOA outperforms the original QAOA 
for the MAX-CUT problem on complete graphs
}

\author{Eunok Bae}
\email{eobae@hanyang.ac.kr}
\affiliation{Department of Mathematics, Research Institute for Natural Sciences, Hanyang University, Seoul 04763, Korea}
\author{Soojoon Lee}
\email{level@khu.ac.kr}
\affiliation{Department of Mathematics, Kyung Hee University, Seoul, 02447, Republic of Korea}

\date{\today}

\begin{abstract}
 Quantum approximate optimization algorithms are hybrid quantum-classical variational algorithms designed to approximately solve combinatorial optimization problems such as the MAX-CUT problem. In spite of its potential for near-term quantum applications, it has been known that quantum approximate optimization algorithms have limitations for certain instances to solve the MAX-CUT problem, at any constant level $p$. Recently, the recursive quantum approximate optimization algorithm, which is a non-local version of quantum approximate optimization algorithm, has been proposed to overcome these limitations. However, it has been shown by mostly numerical evidences that the recursive quantum approximate optimization algorithm outperforms the original quantum approximate optimization algorithm for specific instances.
 In this paper, we analytically prove that the recursive quantum approximate optimization algorithm is more competitive than the original one to solve the MAX-CUT problem for complete graphs with respect to the approximation ratio.
\end{abstract}

%\pacs{}
\maketitle
 
 \section{Introduction}
 \label{sec:intro}
 There has been a growing interest in practical quantum computing in the noisy intermediate-scale quantum (NISQ) era. The NISQ devices have several restrictions due to noise in quantum gates and limited quantum resources~\cite{Pre18}. 
 Diverse disciplines, for instances, combinatorial optimization, quantum chemistry, and machine learning, are regarded as potential areas of application to demonstrate a quantum advantage over the best known classical methods in the NISQ devices. 
 
 Quantum approximate optimization algorithm (QAOA) was designed to solve hard combinatorial optimization problems such as the MAX-CUT problem~\cite{FGG14}.
 QAOA is a hybrid quantum-classical algorithm consisting of a parametrized quantum circuit and a classical optimizer to train it, and it has been proposed as one of the principal approaches to address the restrictions of the NISQ devices since the parameters such as the circuit depth can be handled~\cite{FGG14}.
 In particular, QAOA at low levels such as the level-1 QAOA is more suitable for the NISQ algorithms since QAOA at low levels executes small depth quantum circuit while high level algorithms can produce uncorrectable errors on the NISQ devices.
 However, it has been known that QAOA at any constant level has limited performance to solve the MAX-CUT problem on several instances~\cite{Has19, BKKT19, FGG20, Mar21, BM22}.
 %%%%%%%%%%%%%%%%%%%%%%%%%%%%%%%%%%%%%%%%%%%%%%%%%%
 %%추가 (classical alg 보다 안 좋은 예 언급 - reference)
 %%%%%%%%%%%%%%%%%%%%%%%%%%%%%%%%%%%%%%%%%%%%%%%%%%
 
 % RQAOA_advantage
 The recursive QAOA~\cite{BKKT19}, 
 the RQAOA for short, has been recently proposed to overcome the limitations of QAOA. 
However, very few results on the RQAOA have been known~\cite{BKKT19, BGGS21, BKKT22}. 
Moreover, while it was analytically proved in one of them that the level-1 RQAOA performs better than any constant level QAOA for solving the MAX-CUT problem on cycle graphs~\cite{BKKT19}, 
 the others have given only numerical evidences to claim similar arguments for finding the largest energy of Ising Hamiltonian~\cite{BGGS21} and for graph coloring problem~\cite{BKKT22}. 

 % motivation of our work
It is natural to ask whether the RQAOA can always perform better than the original QAOA. We may expect a positive answer because the RQAOA executes QAOA recursively as its subroutine. The simplest way to check the answer would be to compare directly the performance of these algorithms with respect to the approximation ratio for various instances. 
% In particular, we focus on complete graphs in this work. 

It has been known that 
QAOA has the limitations for solving the MAX-CUT problem on simple graphs with large girth and triangle-free regular graphs with constant degree~\cite{WHJR18, Has19, WL20, Mar21}. However, complete graphs are not included in these cases since they have small girth and they are not triangle-free graph with non-constant degree (depends on the number of vertices). 
Thus we focus on complete graphs in this work, 
and show that the original QAOA still has a limited performance for solving the MAX-CUT problem on even complete graphs 
although we can find the solution by intuition.

 % our result

 In this paper, we compare the performance of the original QAOA and the RQAOA for solving the MAX-CUT problem on complete graphs with $2n$ vertices, and show that the approximation ratio of the level-1 RQAOA is exactly one whereas that of the level-1 QAOA is strictly less than $1-1/8n^2$. 
 This implies that the exact solution for the MAX-CUT problem on complete graphs can be found by using the level-1 RQAOA while it can never be done by performing the level-1 QAOA.
 
 In addition, if the level of these algorithms is higher 
 or the order and size of the graph are larger, 
 then the error rate also becomes larger because the circuit depth increases. 
 %%%%%
 Especially, since the complete graph has the largest graph size among simple graphs with the same order, it becomes more difficult to solve this problem with higher level algorithms than the level-1 algorithm in NISQ devices.
 In other words, it is hard to find the solution of the MAX-CUT problem on complete graphs if we use high level QAOA due to uncorrectable errors.
 Therefore, our result implies that the RQAOA could be a better algorithm than the QAOA for the NISQ devices.
 
 % preview
 This paper is organized as follows. 
 In Sec.~\ref{sec:qaoa}, we briefly review the MAX-CUT problem and QAOA to solve it. 
 In Sec.~\ref{sec:analysis}, we introduce the RQAOA which is the non-local variant of QAOA, and prove that the {level-1} RQAOA outperforms the level-1 QAOA for solving the MAX-CUT problem on complete groups. In Sec.~\ref{sec:conclusion}, we summarize our result, and discuss its meaning and related works.
 
\section{Preliminaries: QAOA and RQAOA}
 \subsection{QAOA for the MAX-CUT problem}
 \label{sec:qaoa}
 
 Let $G=(V,E)$ be a graph with the set of vertices $V=\{1,2,\dots,n\}$ and the set of edges $E=\{(i,j):i,j \in V\}$. The MAX-CUT problem is a well-known combinatorial optimization problem which aims to split $V$ into two disjoint subsets such that the number of edges spanning the two subsets
 is maximized. The MAX-CUT problem can be formulated by maximizing the cost function 
 \[
 C(\mathbf{x})=\frac{1}{2} \sum_{(i,j)\in E} \left(1-x_i x_j \right)
 \]
 for $\mathbf{x}=(x_1,x_2,\dots,x_n) \in \{-1,1\}^n$. 
 This classical cost function can be converted to a quantum problem Hamiltonian 
 \[
 H_C=\frac{1}{2} \sum_{(i,j)\in E} \left( I-Z_i Z_j \right),
 \]
 where $Z_i$ is the Pauli operator $Z$ acting on the $i$-th qubit.
 The $p$-level QAOA, denoted by QAOA$_p$, for the MAX-CUT problem can be described as 
 the following algorithm.
 \begin{algorithm}[QAOA$_{p}$~\cite{FGG14}]
The QAOA$_{p}$ is as follows.
 \begin{enumerate}
     \item Initialize the quantum processor in $\ket{+}^{\otimes n}$.
     \item Generate a variational wave function
     \[\ket{\psi_p(\bm{\beta},\bm{\gamma})}
     =e^{-i\beta_pH_B}e^{-i\gamma_pH_C}
     \cdots e^{-i\beta_1H_B}e^{-i\gamma_1H_C}
     \ket{+}^{\otimes n},
     \]
     where $\bm{\beta}=(\beta_1,\beta_2,\ldots,\beta_p)$,
     $\bm{\gamma}=(\gamma_1,\gamma_2,\ldots,\gamma_p)$, 
     $H_B=\sum_{i=1}^n X_i$ is a mixing Hamiltonian, 
     and $X_i$ is the Pauli operator $X$ acting on the $i$-th qubit.
     \item Compute the expectation value 
     \[F_p(\bm{\beta},\bm{\gamma})=\bra{\psi_p(\bm{\beta},\bm{\gamma})}H_C\ket{\psi_p(\bm{\beta},\bm{\gamma}})\] 
     by performing the measurement in the computational basis.
     \item Find the optimal parameters 
     \[(\bm{\beta}^*,\bm{\gamma}^*)= \mathrm{argmax}_{\bm{\beta},\bm{\gamma}} F_p(\bm{\beta},\bm{\gamma})\]
     using a classical optimization algorithm.
 \end{enumerate}
 \end{algorithm}
  
The approximation ratio $r$ of QAOA$_{p}$ is defined as \[r=\frac{F_p(\bm{\beta}^*,\bm{\gamma}^*)}{C_{\max}},\] 
where $C_{\max}=\max_{\mathbf{x}\in \{-1,1\}^n}C(\mathbf{x})$.

 \subsection{RQAOA}
 \label{sec:analysis}

In this section, we briefly review the concept of the RQAOA. 
For the level-$p$ RQAOA, denoted by RQAOA$_p$, we consider an Ising-like Hamiltonian 
$$H_n=\sum_{(i,j) \in E}J_{i,j}Z_iZ_j$$ 
which is defined on a graph $G_n=(V,E)$ with $|V|=n$, where $J_{i,j}\in \mathbb{R}$ are arbitrary. The RQAOA$_{p}$ attempts to approximate 
$$\max_{\mathbf{x} \in \{-1,1\}^n} \bra{\mathbf{x}}H_n\ket{\mathbf{x}},$$ 
where 
$
\ket{\mathbf{x}}=\ket{x_1,\dots,x_n}
$, 
$Z\ket{x_i}=x_i\ket{x_i}$ for each $i=1,\dots,n$, 

and it can be described by the following algorithm.

\begin{algorithm}[RQAOA$_{p}$~\cite{BKKT19}]
The RQAOA$_{p}$ consists of the following steps.
\begin{enumerate}
    \item Apply the original QAOA to find the optimal state $\ket{\psi_p(\bm{\beta}^*,\bm{\gamma}^*)}$ which maximizes \[\bra{\psi_p(\bm{\beta},\bm{\gamma})}H_n\ket{\psi_p(\bm{\beta},\bm{\gamma})}.\]
    \item Compute $$M_{i,j}=\bra{\psi_p(\bm{\beta}^*,\bm{\gamma}^*)}Z_iZ_j\ket{\psi_p(\bm{\beta}^*,\bm{\gamma}^*)}$$ for every edges $(i,j)\in E$.
    \item Choose a pair $(k,l)$ which maximizes the magnitude of $M_{i,j}$
    \item Impose the constraint $Z_k=\textrm{sgn}(M_{k,l})Z_l$, and  
    replace it into the Hamiltonian 
    \begin{eqnarray*}
    H_n&=&
    \sum_{(i,k) \in E}J_{i,k}Z_iZ_k +
    \sum_{i,j\neq k}J_{i,j}Z_iZ_j \\
    &=&
    \textrm{sgn}(M_{k,l})\left[\sum_{(i,k) \in E}J_{i,k}Z_iZ_l 
    \right] +    \sum_{i,j\neq k}J_{i,j}Z_iZ_j
    \end{eqnarray*}
    \item Call the RQAOA recursively to maximize the expected value of a new Ising Hamiltonian $H_{n-1}$ depending on $n-1$ variables:
    \[
    H_{n-1}=\sum_{(i,l) \in E'_0}J'_{i,j}Z_iZ_l +\sum_{(i,j) \in E'_1}J'_{i,j}Z_iZ_j, 
    \]
    where  
    \begin{eqnarray*}
    E'_0&=&\{(i,l) : (i,k)\in E\},\\
    E'_1&=&\{(i,j) : i,j \neq k\}, 
    \end{eqnarray*}
    and 
    \[
    J'_{i,j}=
    \begin{cases}
    \mathrm{sgn}(M_{k,l})J_{i,k}
     & \mathrm{if}~(i,l)\in E'_0, \\
    J_{i,j} & \mathrm{if}~(i,j)\in E'_1.
    \end{cases}
    \]
    \item The recursion stops when the number of variables reaches some suitable threshold value $n_c \ll n$, and find $\mathbf{x}^*=\mathrm{argmax}_{\mathbf{x}\in \{-1,1\}^{n_c}}\bra{\mathbf{x}}H_{n_c}\ket{\mathbf{x}}$ by a classical algorithm.
    \item Reconstruct the original (approximate) solution $\tilde{\mathbf{x}}\in\{-1,1\}^n$ from $\mathbf{x}^*$ using the constraints.
\end{enumerate}
 \end{algorithm}
 
 \section{Our result}
 \label{sec:result}
 
Now, we investigate the performance of the original QAOA$_{1}$ and RQAOA$_{1}$ for the MAX-CUT problem on complete graphs, and we have the following theorem.

\begin{theorem}
Let $K_{2n}$ be the complete graph with $2n$ vertices for $n\ge2$ and let $H_C=\frac{1}{2} \sum_{(i,j) \in E}\left(I-Z_iZ_j\right)$ be the problem Hamiltonian for the MAX-CUT problem. Then 
\begin{enumerate}
    \item RQAOA$_{1}$ achieves the approximation ratio 1.
    \item The approximation ratio of QAOA$_{1}$ is strictly less than 
    $ 1-\frac{1}{8n^2}$.
\end{enumerate}
\end{theorem}
\begin{proof}
We first show that RQAOA$_{1}$ achieves the approximation ratio 1.
Let $$H_{2n}=\frac{1}{2} \sum_{(i,j) \in E}\left(I-Z_iZ_j\right),$$
where $i,j$ are vertices of $K_{2n}$.
Consider a cost function of the form
$$
C_{2n}(\mathbf{x})=\sum_{(i,j) \in E_{2n}}\left(I-x_ix_j\right),
$$
where $E_{2n}$ denotes the edge set of $K_{2n}$.

Suppose that 
\[
({\beta}^*,{\gamma}^*)=\textrm{argmax}_{\beta,\gamma} \bra{\psi_1({\beta},{\gamma})}H_{2n}\ket{\psi_1({\beta},{\gamma})}.
\]
The exact form for the expectation value for QAOA with $p=1$ has been known in~\cite{WHJR18} and it allows us to calculate $M_{ij}$ as follows.
For each edge $(i,j)$,
\begin{eqnarray*}
M_{ij} &=& \bra{\psi_{  1}({ {\beta^*},{\gamma^*}})}Z_iZ_j 
\ket{\psi_{  1}({ {\beta^*},{\gamma^*}})} \\
&=& \frac{1}{4} \sin4\beta^* \cdot \sin\gamma^* \cdot 2\cos^{2n-2}\gamma^*  \\
&& - \frac{1}{4}\sin^22\beta^* \left( 1-\cos^{2n-2}2\gamma^* \right)  \\
&=& \frac{1}{2} \sin4\beta^* \cdot \sin\gamma^* \cdot \cos^{2n-2}\gamma^* \\ &&-\frac{1}{8}\left(1-\cos4\beta^*\right)
\left[1-\left( 2\cos^2\gamma^* -1 \right)^{n-1}
\right].
\end{eqnarray*}
For the recursion step, we can pick a pair $(k,l)$ in $E$ randomly since all $M_{ij}$'s coincide. Without loss of generality, assume that $(k,l)=(2n-1,2n)$.
It can be easily shown that $M_{i,j}<0$ for all edges $(i,j)$ and thus, by imposing the constraint
\begin{equation}
\label{eq:constraint1}
x_{2n}=-x_{2n-1},
\end{equation}
the RQAOA removes the variable $x_{2n}$ from the cost function $C_{2n}(\mathbf{x})$, we obtain the new cost function 
\begin{eqnarray}
C'_{2n}(\mathbf{x}')
&=&\frac{1}{2}|E_{2n}|-\frac{1}{2} \left(
x_1x_{2n}+ \cdots 
+x_{2n-1}x_{2n}\right) \nonumber \\
&&-\frac{1}{2}\sum_{(i,j)\in E_{2n-1}}x_ix_j\nonumber \\
&=&\frac{1}{2}|E_{2n}|+\frac{1}{2} \left(
x_1x_{2n-1}+ \cdots 
+x_{2n-1}x_{2n-1}\right) \nonumber \\
&&-\frac{1}{2}\sum_{(i,j)\in E_{2n-1}}x_ix_j\nonumber \\
&=& \frac{1}{2}|E_{2n}|+\frac{1}{2}-\frac{1}{2}
\sum_{(i,j)\in E_{2n-2}}x_ix_j. \nonumber \\
&=& \frac{1}{2}|E_{2n}|+\frac{1}{2}-\frac{1}{2}|E_{2n-2}|+C_{2n-2}(\mathbf{x}'),
\label{eq:C_6}
\end{eqnarray}
where $\mathbf{x}' \in \{-1,1\}^{2n-2}$.

\begin{figure}
%\begin{figure}[h]
	\centering
	\includegraphics[width=\linewidth]{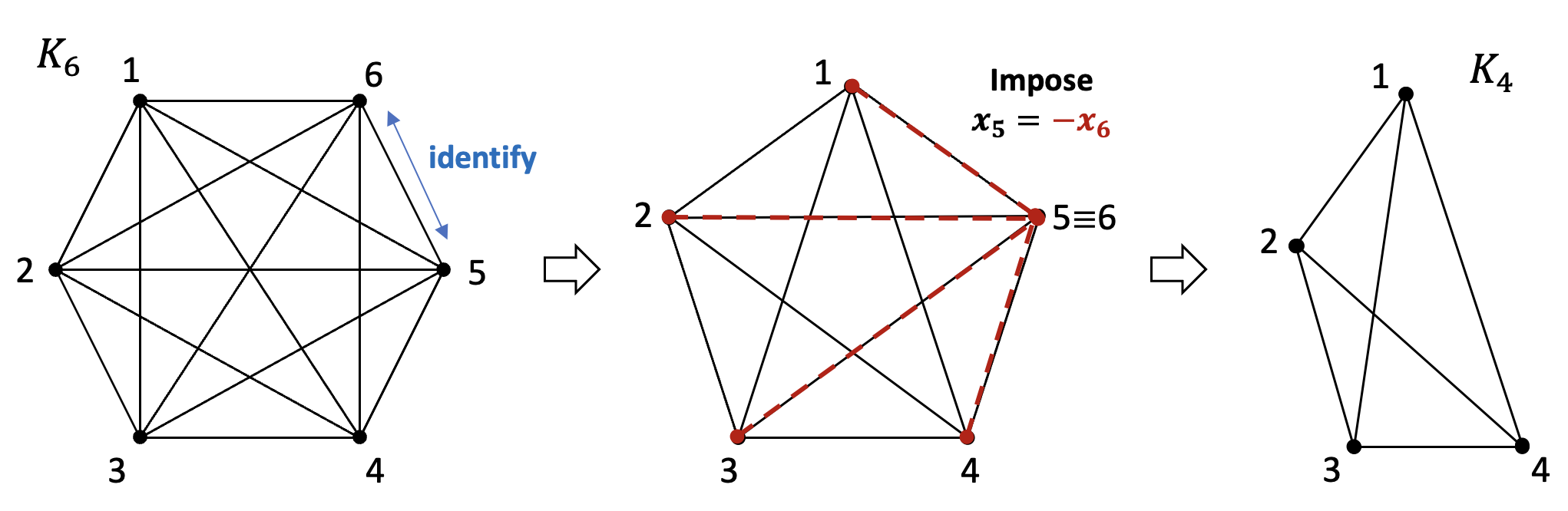}
	\caption{A schematic diagram showing the change of the cost function after one iteration of the RQAOA$_1$ through the graph; The RQAOA$_1$ eliminates the variable $x_6$ in $K_6$ by imposing the constraint $x_5=-x_6$ on the cost function $C_6(\mathbf{x})$. In the middle graph, the red dashed edges indicate the terms including $x_6$ in {  $C_6(\mathbf{x})$,} and these terms 
    {  are cancelled} out after substituting $-x_5$ for $x_6$ due to the different sign. 
	As a consequence, we obtain the new cost function in terms of $C_4(\mathbf{x}')$ with additional terms as we can see in Eq.~(\ref{eq:C_6}).
	}
	\label{Fig1}
\end{figure}

Similarly, the RQAOA eliminates the variable $x_{2n-2}$ by imposing the constraint
\begin{equation}
x_{2n-2}=-x_{2n-3}
\end{equation}
on the cost function $C_{2n}(\mathbf{x})$, and we have the next cost function 
\begin{eqnarray*}
C''_{2n}(\mathbf{x}'')
&=&\frac{1}{2}|E_{2n}|+\frac{1}{2}-\frac{1}{2}
\sum_{(i,j)\in E_{2n-2}}x_ix_j \\
&=& \frac{1}{2}|E_{2n}|+\frac{1}{2}+\frac{1}{2} \left(
1 - \sum_{(i,j)\in E_{2n-4}}x_ix_j \right) \\
&=& \frac{1}{2}|E_{2n}|+\frac{1}{2}+\frac{1}{2} 
-\frac{1}{2} |E_{2n-4}| + C_{2n-4}(\mathbf{x}''),
\end{eqnarray*}
where $\mathbf{x}'' \in \{-1,1\}^{2n-4}$.
By imposing the following $k$ constraints inductively 
\begin{eqnarray}
x_{2n}&=&-x_{2n-1} \nonumber\\ 
x_{2n-2}&=&-x_{2n-3} \nonumber\\
&\vdots& \nonumber \\
x_{2n-(2k-2)}&=&-x_{2n-(2k-1)},
\label{eq:constraint2}
\end{eqnarray}
the cost function $C_{2n}(\mathbf{x})$ after eliminating variables $x_{2n}, x_{2n-2}, \dots,  x_{2n-(2k-2)}$ becomes 
\begin{eqnarray*}
\frac{1}{2}|E_{2n}|+\frac{k}{2}
-\frac{1}{2} |E_{2n-2k}| + C_{2n-2k}(\mathbf{\tilde{x}}),
\end{eqnarray*}
where $\mathbf{\tilde{x}} \in \{-1,1\}^{2n-2k}$.
Now, we observe that
\[
\max_{\mathbf{x} \in \{-1,1\}^{2n}}  C_{2n}(\mathbf{x})
\]
is not less than 
\begin{eqnarray*}
\max_{\mathbf{x} \in \mathcal{X}} C_{2n}(\mathbf{x} ) 
&=& \frac{1}{2}n(2n-1) - \frac{1}{4}(2n-2k)(2n-2k-1) \\
&&+ \frac{k}{2} + 
 \max_{\mathbf{\tilde{x}} \in \{-1,1\}^{2n-2k}} C_{2n-2k}(\mathbf{\tilde{x}}) \\
&=& \frac{1}{2}n(2n-1) - \frac{1}{2}(n-k)(2n-2k-1) \\
&&+ \frac{k}{2} + (n-k)^2 \\
&=& n^2, 
\end{eqnarray*}
where $\mathcal{X}$ is the subset of $\{-1,1\}^{2n}$ satisfying the $k$ constraints in Eqs.~(\ref{eq:constraint2}).
This completes the proof of the {  first} 
statement.

We now show that 
the approximation ratio of QAOA$_{1}$ is strictly less than 1 for every $n$.
In order to obtain the bounds for the approximation ratio of QAOA${_1}$, we take the exact formula in~\cite{WHJR18} once again. For a complete graph with $2n$ vertices and $n \ge 2$, we have 
\begin{eqnarray}
\left< C_{ij} \right>
&=& \frac{1}{2}-\frac{1}{4} \sin^2(2\beta)\left(1-\cos^{2n-2}(2\gamma) \right) 
\nonumber\\
&&+\frac{1}{2}\sin(4\beta)\sin\gamma\cos^{2n-2}(\gamma),
\end{eqnarray}
where $C_{ij}= \frac{1}{2}I-Z_iZ_j$ and 
\[\left< C_{ij} \right> {  =} \bra{\psi_1({ {\beta},{\gamma}})}C_{ij}\ket{\psi_1({ {\beta},{\gamma}})}.\]
The QAOA$_{1}$ for the MAX-CUT problem on the complete graph $K_{2n}$ maximizes the expectation value
\begin{eqnarray*}
F_1({  {\beta}, {\gamma}})
&=& \bra{\psi_1({ {\beta},{\gamma}})}H_C\ket{\psi_1({ {\beta},{\gamma}})} \\
&=& |E_{2n}| \left< C_{ij} \right>,
\end{eqnarray*}
or, equivalently, it maximizes the following function with respect to the parameters {  $\beta$ and $\gamma$}.
\begin{eqnarray*}
f({  \beta,\gamma})&:=&\frac{1}{2}\sin(4\beta)\sin\gamma\cos^{2n-2}(\gamma) \\
&&-\frac{1}{4} \sin^2(2\beta)\left(1-\cos^{2n-2}(2\gamma) \right) \\
&=& \frac{1}{2}\sin(4\beta)\sin\gamma\cos^{2n-2}(\gamma) \\
&&-\frac{1}{8}\left(1-\cos(4\beta) \right)\left(1-\cos^{2n-2}(2\gamma) \right).
\end{eqnarray*}

Let us first differentiate the function $f$ by $\beta$ to obtain the optimal $\beta$ as a function of $\gamma$.
\begin{eqnarray*}
\frac{\partial f}{\partial \beta} &=& 
2 \cos(4\beta)\sin\gamma\cos^{2n-2}(\gamma) \\
&&- \frac{1}{2} \sin(4\beta) \left(1-\cos^{2n-2}(2\gamma) \right).
\end{eqnarray*}
If $\cos^{2n-2}(2\gamma)=1$, then $\cos(2\gamma)=\pm 1$ and so, 
$$
\sin^2(\gamma)=\frac{1-\cos(2\gamma)}{2}=0
$$
or
$$
\cos^2(\gamma)=\frac{1+\cos(2\gamma)}{2}=0
$$
which implies that 
$f({  \beta,\gamma})=0$.

Now, we assume that $\cos^{2n-2}(2\gamma) \ne 1$. Then we have
\begin{equation}
\label{eq:beta}
    \frac{\partial f}{\partial \beta} =0 \iff \tan(4\beta)=\frac{4\sin\gamma\cos^{2n-2}(\gamma)}{1-\cos^{2n-2}(2\gamma)},
\end{equation}
and hence the optimal parameter {  $\beta_\gamma^*$} satisfies
\[
{  \beta_\gamma^*} = \arctan \left(\frac{4\sin\gamma\cos^{2n-2}(\gamma)}{1-\cos^{2n-2}(2\gamma)} \right).
\]
Using the trigonometric identities 
\begin{eqnarray*}
\sin \left(\arctan(x) \right) &=& \frac{x}{\sqrt{1+x^2}}, \\
\cos \left(\arctan(x) \right) &=& \frac{1}{\sqrt{1+x^2}}
\end{eqnarray*}
for $x>0$, we obtain
\begin{eqnarray*}
f({ \beta_\gamma^*},\gamma)&=&
\frac{1}{2}\sin \left( \arctan x(\gamma) \right)\sin\gamma\cos^{2n-2}(\gamma) \\
&&- 
\frac{1}{8}\left[1-\cos\left( \arctan x(\gamma) \right) \right]\left(1-\cos^{2n-2}(2\gamma) \right),
\end{eqnarray*}
where 
\[
x(\gamma)=\frac{4\sin\gamma\cos^{2n-2}(\gamma)}{1-\cos^{2n-2}(2\gamma)}.
\]
For the simplicity of calculation, let $d=2n-2$, $s_{\gamma}=\sin\gamma$, $c_{\gamma}=\cos\gamma$, then
$$
x(\gamma) 
=\frac{4s_{\gamma}c_{\gamma}^d}{1-\left(2c_{\gamma}^2-1\right)^d},
$$
and the function $f({ \beta_\gamma^*},\gamma)$ can be rewritten and simplified by using the constraint in~Eq.~(\ref{eq:beta}) as
\begin{eqnarray*}
f({ \beta_\gamma^*},\gamma)&=&
\frac{1}{2}\frac{x(\gamma)}{\sqrt{1+x(\gamma)^2}}
s_{\gamma}c_{\gamma}^d\\
&&- \frac{1}{8}\left(1-\frac{1}{\sqrt{1+x(\gamma)^2}} \right)\left(1-\left(2c_{\gamma}^2-1\right)^d \right) \\
&=&\frac{1}{8}\frac{x(\gamma)}{\sqrt{1+x(\gamma)^2}}
\left(1-\left(2c_{\gamma}^2-1\right)^d \right) x(\gamma)\\
&&-  \frac{1}{8}\left(1-\frac{1}{\sqrt{1+x(\gamma)^2}} \right)\left(1-\left(2c_{\gamma}^2-1\right)^d \right) \\
&=& \frac{1}{8}\left(1-\left(2c_{\gamma}^2-1\right)^d \right)
\left(\sqrt{1+x(\gamma)^2}-1 \right) \\
&=& \frac{1}{8}\sqrt{\left(1-(2c_{\gamma}^2-1)^d\right)^2 +16s_{\gamma}^2c_{\gamma}^{2d}}\\
&&-\frac{1}{8}\left(1-(2c_{\gamma}^2-1)^d\right).
\end{eqnarray*}

%%%%%%%%%%%%%%%%%%%%%%%%%
% app. ratio < 1
%%%%%%%%%%%%%%%%%%%%%%%%%
Now we let $f(\gamma):=f(\beta_\gamma^*,\gamma)$. Then
we can show that 
\begin{eqnarray}
f(\gamma)
&<& \frac{1}{4n-1}
\label{eq:condition_f}
\end{eqnarray}
for all $\gamma$,
since the inequality in Eq.~(\ref{eq:condition_f}) is equivalent to
the following inequality
\begin{eqnarray}
 \frac{4}{(4n-1)^2}+\frac{1}{4n-1}\left( 1-\left(2c_{\gamma}^2-1\right)^d \right) > (1-c_{\gamma}^2)c_{\gamma}^{2d}, 
\end{eqnarray}
and a function g(t) defined as 
\[
g(t):=\frac{4}{(4n-1)^2}+\frac{1}{4n-1}\left( 1-\left(2t-1\right)^d \right) - (1-t)t^d
\]
can be shown to be strictly greater than zero for all $t\in [0,1]$ 
(See Appendix~\ref{sec:g(t)} for the details).
Therefore, the approximation ratio of QAOA$_1$ for the MAX-CUT problem on complete graphs $K_{2n}$ is
\begin{eqnarray*}
\nonumber
\frac{F_{1}({\beta}^*,{\gamma}^*)}{\max_{\mathbf{x}}C_{2n}(\mathbf{x})} 
&=& \frac{\max_{{\beta},{\gamma}}F_1({\beta},{\gamma})}{n^2} \\
&=& \frac{|E_{2n}|\left(\frac{1}{2}+\max_{\gamma}f(\gamma)\right)}{n^2} 
 \\
&=&\frac{(2n-1)\left(\frac{1}{2}+\max_{\gamma}f(\gamma)\right)}{n} \\
&<& 1-\frac{1}{2n(4n-1)} 
 \\
&<& 1-\frac{1}{8n^2}.
\end{eqnarray*}
This completes the proof of the second statement.
\end{proof}

\section{Conclusion}
\label{sec:conclusion}
In this paper, we have analyzed the performance of the level-1 RAOA and the level-1 QAOA to solve the MAX-CUT problem on complete graphs with $2n$ vertices. 
Moreover, we have proved that the level-1 RQAOA achieves the approximation ratio exactly one, which means that it can always find the exact solution. 
On the other hand, we have shown that the approximation ratio of the level-1 QAOA is strictly less than $1-1/8n^2$ for any $n$. 

One of the most interesting points in this work is that QAOA has a limitation to solve the MAX-CUT problem on complete graphs which %is 
{  are contained in} the cases that we {  can} intuitively know what the maximum cut is. 
Let us consider the following situation. 
Assume that there is a black box which can construct the MAX-CUT Hamiltonian corresponding to {  a given input graph,} and we {  can} perform QAOA based on the Hamiltonian to find the solution. 
Our result ensures that {  when} the given graph {  is} a complete graph, QAOA cannot find the solution exactly while RQAOA can find the exact solution. 
Hence, our {  result provides} us with another analytic evidence demonstrating that QAOA has a limited performance, and 
RQAOA overcomes the limitation of the QAOA in NISQ era.

At this point, can we say that the RQAOA really outperforms the original QAOA to solve the MAX-CUT problem? Can we give its analytical proof? These questions are still open for other instances even {  though} we have already had the positive answer for complete graphs as well as cycle graphs. 
In {  particular, the properties of graphs addressed in the previous recursive step of the RQAOA cannot generally be preserved in the next one.} 
For example, 
{  when performing the RQAOA on a regular graph, 
the graph does not still become a regular one after a recursive step in general. 
However, for complete graphs and cycle graphs 
the properties remain intact 
even in the next {  recursive} step} 
although the order and size of the graphs are reduced. 
{  On this account,} it would be more complicated to analytically calculate the approximation ratio of {  the level-1 RQAOA} 
to solve the MAX-CUT problem on other graphs.

Very recently, a limitation of the RQAOA 
{  was} also known, and the reinforcement learning enhanced {  variant of the} RQAOA, called {  the} RL-RQAOA, was proposed to improve the RQAOA~\cite{PJBD22}. {  The authors in Ref.~\cite{PJBD22}} numerically showed that {  the} RL-RQAOA outperforms the RQAOA on some instances. 
Thus, it {  would be an} 
interesting and important future work to analytically prove that {  the} RL-RQAOA outperforms the RQAOA on certain instances.
 
%
%----------------------------------------%
%            Acknowledgements            %
%----------------------------------------%

\begin{acknowledgements}
E.B. thanks Kunal Marwaha for helpful discussions. 
This work was supported by the National Research Foundation of Korea (NRF)
grant funded by the Ministry of Science and ICT (MSIT) (Grants No. NRF-2020M3E4A1079678 and No. NRF-2022R1C1C2006396).
S.L. acknowledges support from the MSIT, Korea, 
under the Information Technology Research Center support program (Grant No. IITP-2022-2018-0-01402) 
supervised by the Institute for Information and Communications Technology Planning and Evaluation 
and Creation of the Quantum Information Science R\&D Ecosystem (Grant No. 2022M3H3A106307411) 
through the NRF funded by the MSIT.
\end{acknowledgements}

\bibliography{RQAOA}

\clearpage
\appendix
\begin{widetext}

\section{The positivity of the function $g(t)$}
\label{sec:g(t)}
In this appendix, we want to show that for all $t\in [0,1]$, 
$$
g(t):=\frac{4}{(4n-1)^2}+\frac{1}{4n-1}\left( 1-\left(2t-1\right)^{2n-2} \right) - (1-t)t^{2n-2} > 0.
$$

To find the minimum of $g(t)$, we  
{  consider} the equation
\begin{eqnarray}
\label{eq:condition_g}
g'(t) = -\frac{4n-4}{4n-1}(2t-1)^{2n-3}+t^{2n-3}\left(-(2n-2)+(2n-1)t\right)=0.
\end{eqnarray}
Since $g$ is continuous, we need to see that $g(0)>0$, $g(1) >0$, and $g(t^*)>0$ for all critical points $t^* \in \left[0,1\right]$.

Observe that 
$$
g(0)=g(1)=\frac{1}{(4n-1)^2}>0.
$$
{  For any critical point $t^* \in (0,1)$, it is clear that}
\begin{equation}
g'(t^*)=-\frac{4n-4}{4n-1}(2t^*-1)^{2n-3} + {t^*}^{2n-3}\left(-(2n-2)+(2n-1)t^*\right) = 0,
\end{equation}
{  that is,}
\begin{equation}
\label{eq:condition_g'}
\frac{4n-4}{4n-1}(2t^*-1)^{2n-3} =  {t^*}^{2n-3}\left(-(2n-2)+(2n-1)t^*\right).
\end{equation}
Now, by imposing the condition in Eq.~(\ref{eq:condition_g'}) on the function $g$, we have
\begin{eqnarray}
g(t^*)&=&\frac{4}{(4n-1)^2}+\frac{1}{4n-1}\left( 1-\left(2t^*-1\right)^{2n-2} \right) - (1-t^*){t^*}^{2n-2} \nonumber \\
&=& \frac{4}{(4n-1)^2}+\frac{1}{4n-1}
\left[ 1-\frac{4n-1}{4n-4}(2t^*-1){t^*}^{2n-3}(-(2n-2)+(2n-1)t^*) \right] -(1-t^*){t^*}^{2n-2} \nonumber \\
&=& \frac{4}{(4n-1)^2}+\frac{1}{4n-1}
+{t^*}^{2n-3} \left( 
-\frac{1}{2n-2}{t^*}^2 + \frac{2n-1}{4n-4}t^*-\frac{1}{2}
\right) \nonumber \\
&=& \frac{4}{(4n-1)^2}+\frac{1}{4n-1}
+\frac{{t^*}^{2n-3}}{2n-2} \left( 
-\left(t^*-\frac{2n-1}{4} \right)^2
+ \frac{(2n-1)^2}{16}-(n-1)
\right). \label{eq:g}
\end{eqnarray}
If we regard the third term in {  Eq.~(\ref{eq:g})} as a function of {  $t^*$}, we can easily see that it is decreasing on $(0,1)$.
Therefore,
\begin{eqnarray*}
g(t^*) 
&=& \frac{4}{(4n-1)^2}+\frac{1}{4n-1}
+\frac{{t^*}^{2n-3}}{2n-2} \left( 
-\left(t^*-\frac{2n-1}{4} \right)^2
+ \frac{(2n-1)^2}{16}-(n-1)
\right) \\
&>& \frac{4}{(4n-1)^2}+\frac{1}{4n-1}
+\frac{1}{2n-2} \left(
-\left(1-\frac{2n-1}{4} \right)^2
+ \frac{(2n-1)^2}{16}-(n-1)
\right)\\
&=& \frac{4}{(4n-1)^2}+\frac{1}{4n-1}
+\frac{1}{2n-2} \left(-\frac{1}{2} \right)
\\
&=& \frac{4n-13}{4(n-1)(4n-1)^2}
\\
&>& 0
\end{eqnarray*}
{  for all $n \ge 4$. 
Hence} $g(t)>0$ for all $t \in [0,1]$ {  when $n \ge 4$.}

For $n<4$,  
we can prove that $g(t)>0$ on $[0,1]$ 
from direct calculations as follows.
When $n=2$, {  $g(t)$ becomes}
\begin{eqnarray*}
    g(t)&=&\frac{4}{49} + \frac{1}{7}\left( 1-(2t-1)^2\right) - (1-t)t^2 \\
    &=&\frac{1}{49} \left( 49t^3-77t+28t+4\right)>0
\end{eqnarray*}
since {  it can easily be shown} that $g(t)$ is increasing on $[0,1]$ and $g(0)=4>0$.
When $n=3$, {  $g(t)$ becomes}
$$
g(t)=\frac{1}{121} \left( t^5 - 177t^4 +352t^3-264t^2 +88t +4\right)
$$
and we can show that it is concave on $[0,1]$ and $g(0)=g(1)=4>0$.
Thus, $g(t)>0$ on the interval $[0,1]$.

\vfill \clearpage
\end{widetext}

\end{document}